\documentclass[conference]{IEEEtran}
\usepackage{amssymb,amsthm,amsmath,epsfig,latexsym,graphicx,bm,nccmath}
\usepackage{mathtools}

\usepackage{epstopdf}
\usepackage{cite}
\usepackage{amsfonts}
\usepackage[ruled, lined, linesnumbered, commentsnumbered, longend]{algorithm2e}
\usepackage{textcomp}
\usepackage{xcolor}
\def\BibTeX{{\rm B\kern-.05em{\sc i\kern-.025em b}\kern-.08em
		T\kern-.1667em\lower.7ex\hbox{E}\kern-.125emX}}
\usepackage{caption}	
\usepackage{subcaption}
\usepackage[utf8x]{inputenc}
\usepackage{etoolbox}
\usepackage{booktabs}
\usepackage[top=0.71in,bottom=1in,left=0.625in,right=0.625in]{geometry}

\newtheorem{theorem}{\textbf{Theorem}}

\newtheorem{proposition}{\textbf{Proposition}}

\usepackage[skip=0pt]{caption}
\usepackage[font=footnotesize,skip=0pt]{subcaption}
\usepackage{easyReview} 

\makeatletter
\patchcmd{\@maketitle}
{\addvspace{0.5\baselineskip}\egroup}
{\addvspace{-1\baselineskip}\egroup}
{}
{}
\makeatother

\begin{document}
	
	\title{Network Slicing Resource Management in Uplink User-Centric Cell-Free Massive MIMO Systems}
	
	\author{\IEEEauthorblockN{Manobendu Sarker\IEEEauthorrefmark{1} and Soumaya Cherkaoui\IEEEauthorrefmark{2}} 
		\IEEEauthorblockA{Department of Computer and Software Engineering, Polytechnique Montreal, Canada\\
			Email: \{manobendu.sarker\IEEEauthorrefmark{1}, soumaya.cherkaoui\IEEEauthorrefmark{2}\}@polymtl.ca}
	}
	
	\maketitle
	
	\begin{abstract}
		This paper addresses the joint optimization of per-user equipment (UE) bandwidth allocation and UE–access point (AP) association to maximize weighted sum-rate while satisfying heterogeneous quality-of-service (QoS) requirements across enhanced mobile broadband (eMBB) and ultra-reliable low-latency communication (URLLC) slices in the uplink of a network slicing-enabled user-centric cell-free (CF) massive multiple-input multiple-output (mMIMO) system. The formulated problem is NP-hard, rendering global optimality computationally intractable. To address this challenge, it is decomposed into two sub-problems, each solved by a computationally efficient heuristic scheme, and jointly optimized through an alternating optimization framework. We then propose (i) a bandwidth allocation scheme that balances UE priority, spectral efficiency, and minimum bandwidth demand under limited resources to ensure fair QoS distribution, and (ii) a priority‑based UE-AP association assignment approach that balances UE service quality with system capacity constraints. Together, these approaches provide a practical and computationally efficient solution for resource‑constrained network slicing scenarios, where QoS feasibility is often violated under dense deployments and limited bandwidth, necessitating graceful degradation and fair QoS preservation rather than solely maximizing the aggregate sum‑rate. Simulation results demonstrate that the proposed scheme achieves up to $52\%$ higher weighted sum-rate, $140\%$ and $58\%$ higher QoS success rates for eMBB and URLLC slices, respectively, while reducing runtime by up to $97\%$ compared to the considered benchmarks. 
	\end{abstract}
	
	\vspace{-3mm}
	
	\section{Introduction}	
	The sixth-generation (6G) wireless systems must support heterogeneous services with stringent quality-of-service (QoS) demands, primarily enhanced mobile broadband (eMBB) (high-rate, reliable) and ultra-reliable low-latency communication (URLLC) (ultra-low latency, high reliability) \cite{10190330}. Network slicing (NS) addresses this by partitioning physical infrastructure into logical slices, each optimized for specific QoS objectives \cite{8320765}. Cell-free (CF) massive multiple-input multiple-output (mMIMO), a promising 6G architecture \cite{Ngo2017}, eliminates cell boundaries by enabling distributed access points (APs) to cooperatively serve user equipments (UEs), achieving superior spectral efficiency (SE), uniform experience, and enhanced reliability through spatial diversity and user-centric connectivity \cite{Demir2021a}. However, integrating NS into CF mMIMO introduces challenges in resource allocation, UE association, and inter-slice isolation \cite{10114604, 10198770}. 
	
	In the CF mMIMO systems, existing research has predominantly focused on downlink (DL) slicing. For instance, \cite{9435371} proposes multi-tenant resource allocation maximizing infrastructure revenue while preserving service quality, though without dynamic link selection critical for scalability. The work in \cite{10114604} introduces a hierarchical slicing framework leveraging multiple time scales to maximize cross-service utility, while \cite{10198770} exploits spatial diversity to meet heterogeneous URLLC and eMBB requirements under imperfect channel state information. However, they do not address resource allocation in infeasible resource-constrained cases. While uplink (UL) slicing is equally important and poses distinct challenges due to diverse UE demands and limited inter-UE coordination, it has received comparatively less attention in the literature. To date, UL slicing has not yet been explored in the CF mMIMO systems.
	
	Designing UL access is inherently more complex than DL due to a lack of UE coordination and different resource priorities. While DL prioritizes bandwidth and throughput for content delivery, UL emphasizes power management and timely transmission for latency-critical applications. Thus, the stringent QoS requirements imposed by emerging UL-intensive applications (autonomous vehicles, telemedicine, augmented/virtual reality) cannot be met by reusing DL slices, necessitating dedicated UL slicing mechanisms for CF mMIMO systems. However, practical high-density deployments with limited bandwidth frequently violate resource allocation feasibility, where total minimum bandwidth demand exceeds available capacity (further discussed in Section \ref{Per-UE Bandwidth Allocation}). In such scenarios, maintaining service continuity for all UEs while ensuring fair QoS distribution becomes critical to prevent complete service outage. This motivates the development of resource allocation strategies that gracefully degrade service quality under scarcity while preserving system-wide fairness and reliability.
	
	In this work, we consider UL slicing in a NS-enabled CF mMIMO framework to support heterogeneous eMBB and URLLC traffic while improving system performance under practical assumptions such as pilot contamination (PC). To the best of our knowledge, the UL slicing problem in the context of the CF mMIMO systems has not yet been addressed in the literature. The main contributions are: (i) we formulate a joint optimization problem for per-UE bandwidth allocation and UE–AP association to maximize weighted sum-rate in NS-enabled UL CF mMIMO systems; due to the NP-hardness of the problem, we propose an alternating-optimization (AO) approach that decomposes it into two tractable sub-problems solved iteratively, (ii) we develop two computationally efficient heuristic algorithms: an efficiency‑oriented bandwidth allocation scheme and a priority‑based UE–AP association scheme, ensuring fair QoS distribution even when feasibility is violated, thereby enabling graceful degradation under resource scarcity, and (iii) extensive numerical results demonstrate that the proposed solution achieves a favorable trade-off between weighted sum-rate and QoS satisfaction while substantially lowering computational cost compared to an interior-point solver-based solution.
	\vspace{-3mm}
	\section{System Model}
	\subsection{Description of Network Configuration}
	
	We consider the UL of a CF mMIMO system with $M$ APs, each equipped with $N$ antennas, and $K$ single-antenna UEs geographically distributed across the service area. All APs are connected to a central processing unit (CPU) via fronthaul links, which coordinates their joint service of UEs grouped into $S=2$ network slices (eMBB and URLLC). Let $\mathcal{M}=\{1,\dots,M\}$, $\mathcal{K}=\{1,\dots,K\}$, and $\mathcal{S}=\{1,\dots,S\}$ denote the sets of APs, UEs, and slices, respectively, and let $\mathcal{K}_s\subseteq\mathcal{K}$ represent the UEs belonging to slice $s\in\{\mathrm{eMBB},\mathrm{URLLC}\}$. The subset of APs serving UE $k$ is denoted by $\mathcal{V}_k$. Each UE $k$ is associated with a single slice $s_k\in\mathcal{S}$ but may be simultaneously served by multiple APs supporting $s_k$. Each UE $k$ is allocated bandwidth $b_k\ge0$, while each slice $s$ has a bandwidth budget $B_s$ satisfying $\sum_s B_s\le B$, where $B$ is the total available bandwidth. Unlike \cite{10198770}, which restricts APs to specific service types, we allow each AP to serve both eMBB and URLLC UEs to preserve spatial diversity and improve load balance. Each AP can serve at most $\tau_p$ (pilot length) UEs concurrently, and every UE $k$ is assigned a positive priority weight $w_k>0$ reflecting heterogeneous QoS levels.
	
	\vspace{-2mm}
	\subsection{Channel Model}
	\label{Channel Model}
	The channel between UE $k \in \mathcal{K}$ and AP $m \in \mathcal{M}$ is modeled as $\mathbf{g}_{k,m} = \sqrt{\beta_{k,m}} \, \mathbf{h}_{k,m}$, where $\mathbf{h}_{k,m} \in \mathbb{C}^{N \times 1}$ denotes the small-scale fading vector, modeled as $\mathcal{CN}(\mathbf{0}, \mathbf{I}_N)$. The large-scale fading coefficient (LSFC) $\beta_{k,m}$ captures both path loss and shadow fading, and is given as $\beta_{k,m} = 10^{\frac{\text{PL}_{k,m}}{10}} \times 10^{\frac{\sigma_{\mathrm{sh}} z_{k,m}}{10}}$, where $\text{PL}_{k,m}$ is the path loss in dB, $\sigma_{\mathrm{sh}}$ is the standard deviation of shadow fading (in dB), and $z_{k,m} \sim \mathcal{N}(0,1)$ models log-normal shadowing \cite{Ngo2017}. The path loss follows a three-slope propagation model \cite{Ngo2017}. It is assumed that the CPU and all APs have perfect knowledge of the LSFCs $\beta_{k,m}$ for all $(k,m)$ pairs.
	
	\subsection{Channel Estimation Model}
	
	The UL transmission consists of two main stages: training and data transmission. In each coherence interval of length $\tau_c$ symbols, $\tau_p$ symbols are allocated for channel estimation and the rest is for data transmission. During UL training, all UEs transmit their assigned pilot sequences, which are received by the APs for channel estimation. As in practical systems, we assume that $K \gg \tau_p$, implying that some UEs reuse the same pilot sequences, leading to PC. The channel is estimated using the minimum mean square error (MMSE) approach and the quality of the estimated channel $\hat{\mathbf{g}}_{k,m}$ is quantified by the mean-square of the estimate: $\gamma_{k,m} \triangleq \mathbb{E} \left[ \| \hat{\mathbf{g}}_{k,m} \|^2 \right] = \sqrt{\tau_p \rho_p \, \eta_k^p \,} \beta_{k,m} \, c_{k,m}$. The MMSE scaling coefficient $c_{k,m}$ is given by 
	\begin{equation}
		\footnotesize
		c_{k,m} \triangleq \frac{\sqrt{\tau_p \rho_p} \, \beta_{k,m} \sqrt{\eta_k^p}}{\tau_p \rho_p \sum\limits_{j \in \mathcal{K}} \beta_{j,m} \eta_j^p |\boldsymbol{\psi}_k^H \boldsymbol{\psi}_j|^2 + 1},
	\end{equation}
	where $\boldsymbol{\psi}_k \in \mathbb{C}^{\tau_p \times 1}$ denotes the unit-norm pilot sequence assigned to UE $k \in \mathcal{K}$, i.e., $\|\boldsymbol{\psi}_k\|^2 = 1$. Here, $\rho_p$ represents the normalized pilot signal-to-noise ratio (SNR), and $\eta_k^p \in (0,1]$ denotes the pilot power control coefficient of UE $k$ \cite{Ngo2017}.  
	\vspace{-2mm}
	\subsection{Uplink Data Transmission Model}
	During UL transmission, each UE $k \in \mathcal{K}$ sends data symbols to its cooperating APs,\footnote{Each AP receives all UEs’ signals but decodes only those from its associated UEs according to the UE-AP association scheme.} which forward them to the CPU for joint decoding using channel estimates. URLLC UEs, having short packets, are modeled using finite block length theory rather than the Shannon capacity \cite{10114604}. The achievable rate expression for both eMBB and URLLC UEs is given at the top of the next page, where $V_k = 1-(1+\text{SINR}_k)^{-2}$ is the channel dispersion, ${Q}^{-1}(\cdot)$ is the inverse Gaussian Q-function, $\theta$ is the decoding error probability, and $L_k$ is the packet size of UE $k$. The closed-form signal-to-interference-plus-noise ratio (SINR) expression for UE $k$ ($\text{SINR}_k$) is written as \eqref{SINR_closed} using similar steps described in \cite{1193803}, where $\eta_{k}^d \in (0,1]$ denotes the data power control coefficient of UE $k$ and $\rho_d$ is the normalized UL data SNR.
	\begin{figure*}[tb]
		\footnotesize
		\begin{align}
			R_{k}=&\left\lbrace \begin{array}{llll} b_{k} \left(1 - \frac{\tau_p}{\tau_c} \right) \log_2 \left(1 + \text{SINR}_{k} \right), & \forall k \in \mathcal{K}_{\mathrm{eMBB}}, \\ b_{k} \left(1 - \frac{\tau_p}{\tau_c} \right) \left[\log_2 \left(1 + \text{SINR}_{k} \right)-\sqrt{\frac{V_k}{L_k}} \frac{{{Q}^{-1}}\left(\theta \right)}{\ln 2} \right], & \forall k \in \mathcal{K}_{\mathrm{URLLC}}. \end{array} \right.
			\label{rate_eq}
		\end{align}
		\vspace{-6pt}
		\noindent\rule{\textwidth}{1pt}
	\end{figure*}
	
	\begin{figure*}[tb]
		\footnotesize
		\begin{align}
			\text{SINR}_k = \frac{N^2\rho_d\eta_k^d\left(\sum_{m \in \mathcal{V}_k} \ \gamma_{k,m}\right)^2}{N\rho_d\sum_{k'=1}^{K}\eta_{k'}^d\sum_{m \in \mathcal{V}_k} \ \gamma_{k,m}\beta_{k',m}+N^2\rho_d \sum_{k'\neq k}^{K}\eta_{k'}^d|\psi_k^H\psi_{k'}|^2\left(\sum_{m \in \mathcal{V}_k} \ \gamma_{k,m} \sqrt{\frac{\eta_{k'}^p}{\eta_k^p}} \frac{\beta_{k',m}}{\beta_{k,m}}\right)^2+N\sum_{m \in \mathcal{V}_k} \ \gamma_{k,m}}.
			\label{SINR_closed}
		\end{align}
		\vspace{-6pt}
		\noindent\rule{\textwidth}{1pt}
	\end{figure*}
	\vspace{-2mm}
	\subsection{Delay Model}
	\label{Delay Model}
	To capture the latency characteristics of delay-sensitive services such as URLLC, we adopt an $M/M/1$ queuing model as used in \cite{10772596} to characterize the UL queuing and transmission delay experienced by each UE. For UE $k \in \mathcal{K}_{\mathrm{URLLC}}$, let $\lambda_{k}$ denote the average packet arrival rate (in packets per second). The service rate of UE $k$ is $\mu_{k} = R_{k} / L_{k}$ (packet size in bytes). So, under the stability condition $\mu_{k} > \lambda_{k}$ \cite{kleinrock1975queueing}, the total delay experienced by UE $k$, denoted by $D_{k}$, comprises both queuing and transmission delays and is expressed as: $D_{k} = \frac{1}{\mu_{k} - \lambda_{k}} = \frac{1}{\frac{R_{k}}{L_{k}} - \lambda_{k}}, \quad \forall k \in \mathcal{K}_{\mathrm{URLLC}}$.
	
	To satisfy the QoS requirements of latency-critical applications, the total delay must not exceed a predefined maximum tolerable latency threshold $D_{k}^{\mathrm{max}}$. This requirement imposes the following constraint on the achievable UL rate: $\frac{R_{k}}{L_{k}} \geq \lambda_{k} + \frac{1}{D_{k}^{\mathrm{max}}}, \quad \forall k \in \mathcal{K}_{\mathrm{URLLC}}$. This constraint ensures that the combined queuing and transmission delays remain within acceptable bounds for each UE of $\mathcal{K}_{\mathrm{URLLC}}$, thereby preserving the reliability and responsiveness required by mission-critical services.
	\vspace{-3mm}
	\section{Problem Formulation}
	Our objective is to jointly optimize the per-UE bandwidth allocation and UE-AP association assignment for maximizing the weighted sum-rate performance under fixed slice-wise bandwidth $B_s, \ s \in \mathcal{S}$, while satisfying service-specific QoS requirements. The weights $w_k$ reflect heterogeneous service priorities of eMBB and URLLC UEs.
	
	To incorporate UE-AP association as a variable, we define an association matrix $\mathbf{A} \in \{0,1\}^{K \times M}$, where $a_{k,m} = 1$ indicates association between AP $m$ and UE $k$, and $a_{k,m} = 0$ otherwise. Thus, $\mathcal{V}_k = \{m: a_{k,m} = 1\}$ with cardinality $|\mathcal{V}_k| = \sum_{m=1}^M a_{k,m}$. Using these definitions, the SINR expression in (\ref{SINR_closed}) is reformulated as (\ref{SINR_closed_up}), replacing the summations over subsets $\mathcal{V}_k$ with terms involving $\mathbf{A}$.
	
	\begin{figure*}[tb]
		\footnotesize
		\begin{align}
			\text{SINR}_k = \frac{N\rho_d\eta_k^d\left(\sum_{m = 1}^M \ a_{k,m}\gamma_{k,m}\right)^2}{\rho_d\sum_{k'=1}^{K}\eta_{k'}^d\sum_{m=1}^M a_{k,m}\gamma_{k,m}\beta_{k',m}+N\rho_d \sum_{k'\neq k}^{K}\eta_{k'}^d|\psi_k^H\psi_{k'}|^2\left(\sum_{m = 1}^M  a_{k,m}\gamma_{k,m} \sqrt{\frac{\eta_{k'}^p}{\eta_k^p}} \frac{\beta_{k',m}}{\beta_{k,m}}\right)^2+\sum_{m  = 1}^M  a_{k,m}\gamma_{k,m}}.
			\label{SINR_closed_up}
		\end{align}
		\vspace{-6pt}
		\noindent\rule{\textwidth}{1pt}
	\end{figure*}
	
	The joint optimization problem can be expressed as follows:
	\begin{subequations}
		\footnotesize
		\begin{alignat}{2}
			\textbf{P0:}\quad 
			\max_{\substack{\mathbf{A} \in \{0,1\},\\ \{b_k\} \geq 0}} \quad &
			\ \sum_{k \in \mathcal{K}} w_k R_{k}, 
			\label{eq:objective_unified} \\
			\text{s.t.}\quad 
			& \sum_{m \in \mathcal{M}} a_{k,m} \geq 1, 
			&& \forall k \in \mathcal{K},  
			\label{c:ue_connectivity_unified} \\
			& \sum_{k \in \mathcal{K}} a_{k,m} \leq \tau_p, 
			&& \forall m \in \mathcal{M}, 
			\label{c:ap_user_limit_unified} \\
			& \sum_{k \in \mathcal{K}_s} b_{k} \leq B_s, 
			&& \forall s \in \mathcal{S}, 
			\label{c:slice_bandwidth_unified} \\
			& \frac{R_{k}}{L_{k}} \geq \lambda_{k} + \frac{1}{D_k^{\max}}, 
			&& \forall k \in \mathcal{K}_{\mathrm{URLLC}}, 
			\label{c:delay_unified} \\
			& R_{k} \geq R_k^{\min}, 
			&& \forall k \in \mathcal{K}_{\mathrm{eMBB}}. 
			\label{c:min_rate_unified}
		\end{alignat}
	\end{subequations}
	
	\noindent
	The objective function \eqref{eq:objective_unified} maximizes the weighted sum-rate of all UEs. In \textbf{P0}, constraint~\eqref{c:ue_connectivity_unified} guarantees full UL coverage by ensuring each UE is associated with at least one AP. Constraint~\eqref{c:ap_user_limit_unified} limits each AP to at most \( \tau_p \) UEs, capturing hardware, scheduling, and pilot reuse constraints~\cite{Bjornson2019}. Constraint~\eqref{c:slice_bandwidth_unified} enforces the slice-wise bandwidth budget \( B_s \), while~\eqref{c:delay_unified} incorporates a queuing-based delay model (see Section~\ref{Delay Model}) for satisfying URLLC latency requirement $D_k^{\max}$. Finally,~\eqref{c:min_rate_unified} ensures eMBB UEs meet the minimum rate requirement \( R_k^{\min} \) for reliable communication. 
	
	\begin{theorem}
		\label{thm:problem_classification}
		\textbf{P0} is a non-convex mixed-integer nonlinear programming (MINLP) problem that is NP-hard.
	\end{theorem}
		
	\begin{proof}
		See in Appendix \ref{Proof of Theorem 1}.
	\end{proof}
	\vspace{-3mm}	
	\section{Proposed Solution}
	To circumvent the NP-hardness in \textbf{P0}, we propose a decoupling approach where the original problem is divided into two sub-problems. Each sub-problem focuses on addressing a specific aspect, such as per-UE bandwidth allocation and UE-AP association assignment.
	
	\subsection{Per-UE Bandwidth Allocation}
	\label{Per-UE Bandwidth Allocation}
	With fixed association matrix $\mathbf{A}$, the per-UE bandwidth allocation $\{b_k\}, \ {k \in \mathcal{K}}$ is formulated as follows:
	\begin{subequations}	
		\footnotesize
		\begin{alignat}{2}
			\textbf{P1:}\quad 
			\max_{\{b_k\} \geq 0} \quad &
			\ \sum_{k \in \mathcal{K}} w_k R_{k},  \label{P1_obj}\\
			\text{s.t.}\quad &  \eqref{c:slice_bandwidth_unified}\text{--}\eqref{c:min_rate_unified}.
		\end{alignat}
	\end{subequations}
	
	\begin{proposition}
		\label{thm:convexity}
		\textbf{P1} is convex since (i) $ R_{k}$ is concave in $b_k$ for constant $\text{SINR}_k > 0$, and (ii) constraints \eqref{c:slice_bandwidth_unified}\text{--}\eqref{c:min_rate_unified} are linear in $b_k$.
	\end{proposition}

	\begin{proposition}
		\label{prop:feasibility}
		For a given association matrix $\mathbf{A}$, \textbf{P1} is feasible if and only if, 
		\begin{equation}
			\footnotesize
			\sum_{k \in \mathcal{K}_{\mathrm{URLLC}}} 
			\frac{L_k \left(\lambda_k + \tfrac{1}{D_k^{\max}}\right)}
			{\mathrm{SE}_k^{\mathrm{URLLC}}}
			+
			\sum_{k \in \mathcal{K}_{\mathrm{eMBB}}} 
			\frac{R_k^{\min}}{\mathrm{SE}_k^{\mathrm{eMBB}}}
			\leq B_s,
			\label{eq:feasibility_condition_final}
		\end{equation}
		where $\mathrm{SE}_k^{\mathrm{eMBB}} = \left(1 - \tfrac{\tau_p}{\tau_c} \right)\log_2\left(1 + \mathrm{SINR}_{k}\right)$ and 
		$\mathrm{SE}_k^{\mathrm{URLLC}} = \left(1 - \tfrac{\tau_p}{\tau_c} \right)
		\left[\log_2 \left(1 + \mathrm{SINR}_{k} \right)
		- \sqrt{\tfrac{V_k}{L_k}} \tfrac{Q^{-1}(\theta)}{\ln 2} \right]$, 
		provided that $\mathrm{SE}_k^{\mathrm{URLLC}} > 0$ for all $k \in \mathcal{K}_{\mathrm{URLLC}}$.
	\end{proposition}
	
	\textbf{Remark.} Proposition~\ref{prop:feasibility} may be violated when: (i) poor channel conditions cause excessive per-UE bandwidth demands exceeding the slice budget, or (ii) high UE density pushes aggregate minimum bandwidth beyond \(B_s\). Standard convex solvers (e.g., CVX) then declare infeasibility without yielding a feasible allocation. To address this infeasibility, we selectively relax constraints for weak-channel UEs and prioritize high-value UEs through weighted objectives to maintain numerical stability and ensure fair QoS reduction while maximizing system utility. Our proposed heuristic solutions for these scenarios are presented in the following subsection.
	
	\subsubsection{Proposed Per-UE Bandwidth Allocation Scheme}
	\label{Proposed Per-UE Bandwidth Allocation Scheme}
	To secure fair QoS distribution even during the feasibility (Proposition \ref{prop:feasibility}) violation scenario, we develop a computationally efficient per-UE bandwidth allocation scheme operating in three stages: (i) priority-based allocation of URLLC UEs, (ii) best-effort allocation of eMBB UEs, and (iii) proportional distribution of residual bandwidth. The procedure is summarized in Algorithm~\ref{alg:greedy_knapsack}. The key idea is to allocate UEs based on a \emph{bandwidth efficiency} metric, defined as $\zeta_k = \tfrac{w_k \mathrm{SE}_k}{b_k^{\min}}$, which balances UE priority, SE, and minimum bandwidth demand $b_k^{\min}$ under limited resources, ensuring fair QoS distribution.
	
	In \textbf{Stage~1 (Lines~\ref{alg:urllc_sort}--\ref{alg:urllc_end})}, URLLC UEs are processed first due to their stringent delay constraints. These UEs are sorted in descending order of $\zeta_k$ and their bandwidths are allocated greedily as long as their minimum bandwidth requirements fit within the remaining URLLC slice budget. This prioritizes high-efficiency URLLC UEs, enhancing the weighted sum-rate while satisfying reliability constraints. In \textbf{Stage~2 (Lines~10--\ref{alg:embb_end})}, the same bandwidth allocation process is applied to eMBB UEs within their slice budget. Since eMBB traffic is elastic, UEs that cannot be allocated simply receive zero bandwidth, implementing a best-effort service model. In \textbf{Stage~3 (Lines~17--\ref{alg:extra_end})}, any residual slice bandwidth is proportionally distributed among allocated UEs again according to $\zeta_k$, ensuring full resource utilization while preserving the efficiency-based fairness established in the allocation phases.
	
	\subsubsection{Computational Complexity Analysis}
	\label{CC:BW}
	The computational complexity of Algorithm~\ref{alg:greedy_knapsack} is primarily determined by the sorting operations in Lines~\ref{alg:urllc_sort} and~\ref{alg:embb_sort}. Sorting $|\mathcal{K}_{\text{URLLC}}|$ and $|\mathcal{K}_{\text{eMBB}}|$ UEs requires $\mathcal{O}(K \log K)$ time. Algorithm~\ref{alg:greedy_knapsack} processes URLLC and eMBB UEs in the first two stages, followed by bandwidth allocation for any remaining resources in the final stage. The admission loops (Lines~\ref{alg:urllc_loop}--\ref{alg:embb_end}) and the residual bandwidth distribution (Lines~\ref{alg:extra_start}--\ref{alg:extra_end}) each operate in $\mathcal{O}(K)$ time. Consequently, the overall computational complexity of the algorithm is $\mathcal{O}(K \log K)$.
	
	\begin{algorithm}[t]
		\caption{Proposed Per-UE Bandwidth Allocation}
		\label{alg:greedy_knapsack}
		\footnotesize
		\SetAlgoLined
		\KwIn{$\{w_k, \text{SE}_k, b_k^{\min}\}$, $\{B_s\}$, $\mathcal{K}_{\mathrm{URLLC}}, \mathcal{K}_{\mathrm{eMBB}}$}
		\KwOut{$\{b_k\}$}
		
		\textbf{Initialize:} $b_k \gets 0, \; \mathcal{A}_s \gets \emptyset, \; B_s^{\text{rem}} \gets B_s, \; \forall s$ \;
		
		$\zeta_k \gets w_k \mathrm{SE}_k / b_k^{\min}, \; \forall k$ \label{alg:efficiency}\;
		
		\label{alg:urllc_start}
		Sort $\mathcal{K}_{\text{URLLC}}$ by $\zeta_k$ (descending) $\to \pi_{\text{URLLC}}$ \label{alg:urllc_sort}\;
		\ForEach{$k \in \pi_{\text{URLLC}}$}{\label{alg:urllc_loop}
			\If{$B_{s_{\text{URLLC}}}^{\text{rem}} \geq b_k^{\min}$}{\label{alg:urllc_check}
				$b_k \gets b_k^{\min}, \; \mathcal{A}_s \gets \mathcal{A} \cup \{k\}$ \; 
				$B_{s_{\text{URLLC}}}^{\text{rem}} \gets B_{s_{\text{URLLC}}}^{\text{rem}} - b_k^{\min}$ \label{alg:urllc_admit}\;
			}
		}
		\label{alg:urllc_end}
		
		\label{alg:embb_start}
		Sort $\mathcal{K}_{\text{eMBB}}$ by $\zeta_k$ (descending) $\to \pi_{\text{eMBB}}$ \label{alg:embb_sort}\;
		\ForEach{$k \in \pi_{\text{eMBB}}$}{
			\If{$B_{s_{\text{eMBB}}}^{\text{rem}} \geq b_k^{\min}$}{
				$b_k \gets b_k^{\min}, \; \mathcal{A}_s \gets \mathcal{A} \cup \{k\}$ \; 
				$B_{s_{\text{eMBB}}}^{\text{rem}} \gets B_{s_{\text{eMBB}}}^{\text{rem}} - b_k^{\min}$ \;
			}
		}
		\label{alg:embb_end}
		
		\label{alg:extra_start}
		\ForEach{$s \in \mathcal{S}$}{\label{alg:extra_loop}
			$\mathcal{A}_s \gets \{k \in \mathcal{A} : s_k = s\}$ \;
			\If{$B_s^{\text{rem}} > 0 \land |\mathcal{A}_s| > 0$}{
				$\zeta_s^{\text{total}} \gets \sum_{k \in \mathcal{A}_s} \zeta_k$ \;
				$b_k \gets b_k + B_s^{\text{rem}} \cdot \zeta_k / \zeta_s^{\text{total}}, \; \forall k \in \mathcal{A}_s$ \label{alg:extra_prop}\;
			}
		}
		\label{alg:extra_end}
	\end{algorithm}
	
	
	\subsection{UE-AP Association Assignment}
	
	Given fixed per-UE bandwidth allocations ${b_k}$, the UE-AP association assignment sub-problem is formulated as follows:
	\begin{subequations}
		\small
		\begin{alignat}{2}
			\textbf{P2:}\quad 
			\max_{\mathbf{A} \in \{0,1\}} \quad &
			\ \sum_{k \in \mathcal{K}} w_k R_{k}, \label{P2_obj}\\
			\text{s.t.}\quad &  \eqref{c:ue_connectivity_unified}\text{--}\eqref{c:ap_user_limit_unified}.
		\end{alignat}
	\end{subequations}
	
	Since SINR depends on $\mathbf{A}$ through complex channel combining and interference coupling, \textbf{P2} is a combinatorial optimization problem that is an MINLP problem. Due to the significant computational burden associated with solving MINLP problems optimally (i.e., $\mathcal{O}(2^{KM})$), a heuristic procedure is developed as described below.
	
	\subsubsection{Proposed UE-AP Association Assignment Scheme}
	\label{Proposed UE-AP Association Assignment Scheme}
	Our proposed scheme constructs the association matrix $\mathbf{A}$ using a priority-based assignment approach that balances UE service quality with system capacity constraints, as outlined in Algorithm~\ref{alg:greedy_association_simple}. To this end, the algorithm identifies the associations that contribute most to the overall objective~\eqref{P2_obj} using a metric called the \textit{association potential}, defined for each UE-AP pair $(k,m)$ as $\Xi_{k,m} = w_k \cdot b_k \cdot \beta_{k,m}$ (Line~\ref{alg:assoc_simple_metric}). UEs are then prioritized in descending order according to $\Pi_k = w_k b_k$ (Lines~\ref{alg:assoc_simple_priority}--\ref{alg:assoc_simple_sort}), ensuring that high-priority UEs (those with larger weights and bandwidth allocations) are served first and have access to their most favorable APs. Next, for each UE $k$, APs are sorted by $\Xi_{k,m}$ (Line~\ref{alg:assoc_simple_ap_sort}), and the algorithm assigns the UE to the AP $m$ with the highest potential, provided the current AP load $\ell_m$ does not exceed the capacity limit $\tau_p$ (Line~\ref{alg:assoc_simple_capacity_check}), thereby satisfying constraint~\eqref{c:ap_user_limit_unified}. If no AP with available capacity can accommodate UE $k$, Algorithm~\ref{alg:greedy_association_simple} performs an {additional assignment} (Lines~\ref{alg:assoc_simple_emergency}--\ref{alg:assoc_simple_force}) to connect UE $k$ to its best available AP $m^*$, thereby guaranteeing minimum connectivity and enforcing constraint~\eqref{c:ue_connectivity_unified}. 
	
	\subsubsection{Computational Complexity Analysis}
	\label{CC:Association}
	The computational complexity of Algorithm~\ref{alg:greedy_association_simple} is dominated by the AP sorting operations in the assignment loop (Lines~9--21). However, computing the association potentials $\Xi_{k,m}$ (Line~\ref{alg:assoc_simple_metric}) requires $\mathcal{O}(KM)$ operations, sorting UEs by priority (Lines~\ref{alg:assoc_simple_priority}--\ref{alg:assoc_simple_sort}) requires $\mathcal{O}(K \log K)$, and sorting $M$ APs for each of the $K$ UEs in the assignment loop incurs $\mathcal{O}(KM \log M)$. Therefore, the overall computational complexity is $\mathcal{O}(KM + K \log K + KM \log M) = \mathcal{O}(KM \log M)$.
	
	\begin{algorithm}[t]
		\caption{Proposed UE-AP Association Assignment Scheme}
		\label{alg:greedy_association_simple}
		\footnotesize
		\SetAlgoLined
		\KwIn{$\{\beta_{k,m}\}$, $\{w_k\}$,  $\{b_k\}$,  $\tau_p$, $\mathcal{K}$, $\mathcal{M}$}
		\KwOut{$\mathbf{A}$}
		\textbf{Initialize:} $\mathbf{A} \gets \mathbf{0}_{K \times M}, \quad \ell_m \gets 0, \; \forall m \in \mathcal{M}$
		
		\label{alg:assoc_simple_potential_start}
		\ForEach{$k \in \mathcal{K}$}{\label{alg:assoc_simple_user_loop}
			\ForEach{$m \in \mathcal{M}$}{\label{alg:assoc_simple_ap_loop}
				$\Xi_{k,m} \gets w_k \cdot b_k \cdot \beta_{k,m}$ \label{alg:assoc_simple_metric} 
			}
		}
		\label{alg:assoc_simple_potential_end}
		
		\label{alg:assoc_simple_priority_start}
		$\Pi_k \gets w_k \cdot b_k, \quad \forall k \in \mathcal{K}$ \label{alg:assoc_simple_priority}\;
		$\mathcal{K}_{\text{sorted}} \gets \text{argsort}(\{\Pi_k\}, \text{descending})$ \label{alg:assoc_simple_sort}\;
		\label{alg:assoc_simple_priority_end}
		
		\label{alg:assoc_simple_init}
		
		\label{alg:assoc_simple_greedy_start}
		\ForEach{$k \in \mathcal{K}_{\text{sorted}}$}{\label{alg:assoc_simple_outer_loop}
			$\mathcal{M}_k^{\text{sorted}} \gets \text{argsort}(\{\Xi_{k,m}\}_{m=1}^M, \text{descending})$\; \label{alg:assoc_simple_ap_sort} 
			$n_{\text{assigned}} \gets 0$ \;
			
			\ForEach{$m \in \mathcal{M}_k^{\text{sorted}}$}{\label{alg:assoc_simple_inner_loop}
				\If{$\ell_m < \tau_p$}{\label{alg:assoc_simple_capacity_check}
					$a_{k,m} \gets 1$ \label{alg:assoc_simple_assign}\;
					$\ell_m \gets \ell_m + 1$; $n_{\text{assigned}} \gets n_{\text{assigned}} + 1$ \;

				}
			}
			
			\If{$n_{\text{assigned}} = 0$}{\label{alg:assoc_simple_emergency}
				$m^* \gets \mathcal{M}_k^{\text{sorted}}(1)$; 
				$a_{k,m^*} \gets 1, \; \ell_{m^*} \gets \ell_{m^*} + 1$; \label{alg:assoc_simple_force} 
			}
		}
		\label{alg:assoc_simple_greedy_end}
	\end{algorithm}
	
	\subsection{The Overall Solution to the Problem \textbf{P0}}
	\label{The Overall Solution to the Problem P0}
	
	To solve the coupled optimization problem \textbf{P0}, we adopt an AO framework that iteratively solves the two sub-problems until convergence. Starting from an initial association matrix $\mathbf{A}^{(0)}$, where each UE connects to the AP with the highest channel strength $\{\beta\}$, the framework alternates between: (i) solving \textbf{P1} for per-UE bandwidth allocation ${b_k}$ (Algorithm~\ref{alg:greedy_knapsack}) with fixed $\mathbf{A}$, and (ii) solving \textbf{P2} for association matrix $\mathbf{A}$ (Algorithm~\ref{alg:greedy_association_simple}) with fixed ${b_k}$. The process stops when the objective~\eqref{eq:objective_unified} improvement is below $\epsilon$ or the iteration count reaches $I_{\max}$. Convergence is guaranteed as each iteration monotonically improves a bounded objective, typically within 3--5 iterations. The per-iteration complexity is $\mathcal{O}(K \log K + KM \log M)$, as discussed in Sections~\ref{CC:BW} and~\ref{CC:Association}, ensuring scalability and computational efficiency.
	\vspace{-2mm}	
	\section{Numerical Analysis}
	\label{Numerical Analysis}
	
	The proposed algorithms are evaluated in a UL CF mMIMO system deployed over a $1 \times 1$ km$^2$ area with $M=100$ APs (each having $N = 4$ antennas) serving $K$ single-antenna UEs. Both APs and UEs are uniformly distributed, and the wrap-around technique \cite{Bjornson2017} is applied to eliminate boundary effects. Following \cite{Ngo2017}, we set $\tau_p = 10$, $\tau_c= 200$, $\rho_p=\rho_d= 100~\text{mW/Np}$, and adopt log-normal shadow fading with standard deviation $\sigma_{sh}= 8$ dB. To induce PC effects more, pilots are randomly assigned to UEs  and their power is controlled using \cite{Sark2301:Pilot}, while UL data transmission employs open-loop power control. The maximum transmit power for both pilot and data signals is set to 100 mW.
	
	To meet eMBB and URLLC QoS demands, the total bandwidth $B = 80$ MHz is evenly split among the slices ($B_s = 40$ MHz each). The UE population includes 40$\%$ eMBB and 60$\%$ URLLC UEs. URLLC traffic follows dynamic arrivals with packet sizes $L_k \in [20,120]$ bytes, rates $\lambda_k \in [5,25]$ packets/s, delay limits $D_k^{\max} \in [0.5,2.5]$ ms, and priority weights $w_k \in [2,4]$, with decoding error probability $\theta = 10^{-5}$. eMBB UEs include 30$\%$ premium ($R_k^{\min} \in [5,10]$ Mbps, $w_k=1.5$) and 70$\%$ standard ($R_k^{\min} \in [1,3]$ Mbps, $w_k=1.0$). Finally, the AO algorithm stops when the objective improvement is below convergence threshold $\epsilon = 0.001$ or iteration limit $I_{\max}=15$ is reached.   
	
	\subsubsection{Comparison of Average Weighted Sum-rate Performance}
	\label{Comparison of Average Weighted Sum-rate Performance}
	
	Fig.~\ref{Average_sumrate} compares the average weighted sum-rate of the proposed scheme (`Proposed') with a `Hybrid' scheme combining CVX-based per-UE bandwidth allocation ($\mathcal{O}(K^3)$~\cite{boyd2004convex}) and UE-AP association from~\cite{sarker2023access} ($\mathcal{O}(KM)$), and a `Baseline' scheme using the same association with round-robin bandwidth allocation. However, CVX can fail due to Proposition~\ref{prop:feasibility} violation, then the Hybrid scheme applies a fallback allocation ($\mathcal{O}(K)$) that greedily distributes $B_s$ among UEs, ensuring graceful degradation under infeasible QoS conditions. All schemes, except the Baseline, employ the AO framework (described in Section~\ref{The Overall Solution to the Problem P0}) to solve \textbf{P0}.
	
	As expected, the Hybrid scheme achieves the highest weighted sum-rate, outperforming the Proposed scheme by approximately $23\%$, owing to its use of the optimal CVX solver for bandwidth allocation (with greedy fallback in infeasible cases). However, this performance gain comes at the expense of significantly higher computational complexity ($\mathcal{O}(K^3 + KM)$ compared to $\mathcal{O}(K \log K + KM \log M)$) and the need for numerical solvers, making the Proposed scheme more suitable for resource-constrained or real-time deployments. However, the Proposed scheme outperforms the Baseline by up to $52\%$. This performance gain is attributed to: (i) the proposed bandwidth allocation (Algorithm~\ref{alg:greedy_knapsack}), which leverages the metric $\zeta_{k}$ capturing UE priority, channel quality, and resource efficiency, and redistributes leftover bandwidth to improve adaptability even when QoS requirements cannot be met, and (ii) the priority-based UE-AP association (Algorithm~\ref{alg:greedy_association_simple}) based on the association potential metric $\Xi_{k,m}$, which complements the per-UE bandwidth allocation. At higher UE densities, the weighted sum-rate performance of all schemes diminishes due to increased inter-UE interference and bandwidth constraints.
	
	\begin{figure}[tb]	
		\centering
		\includegraphics[scale=0.2]{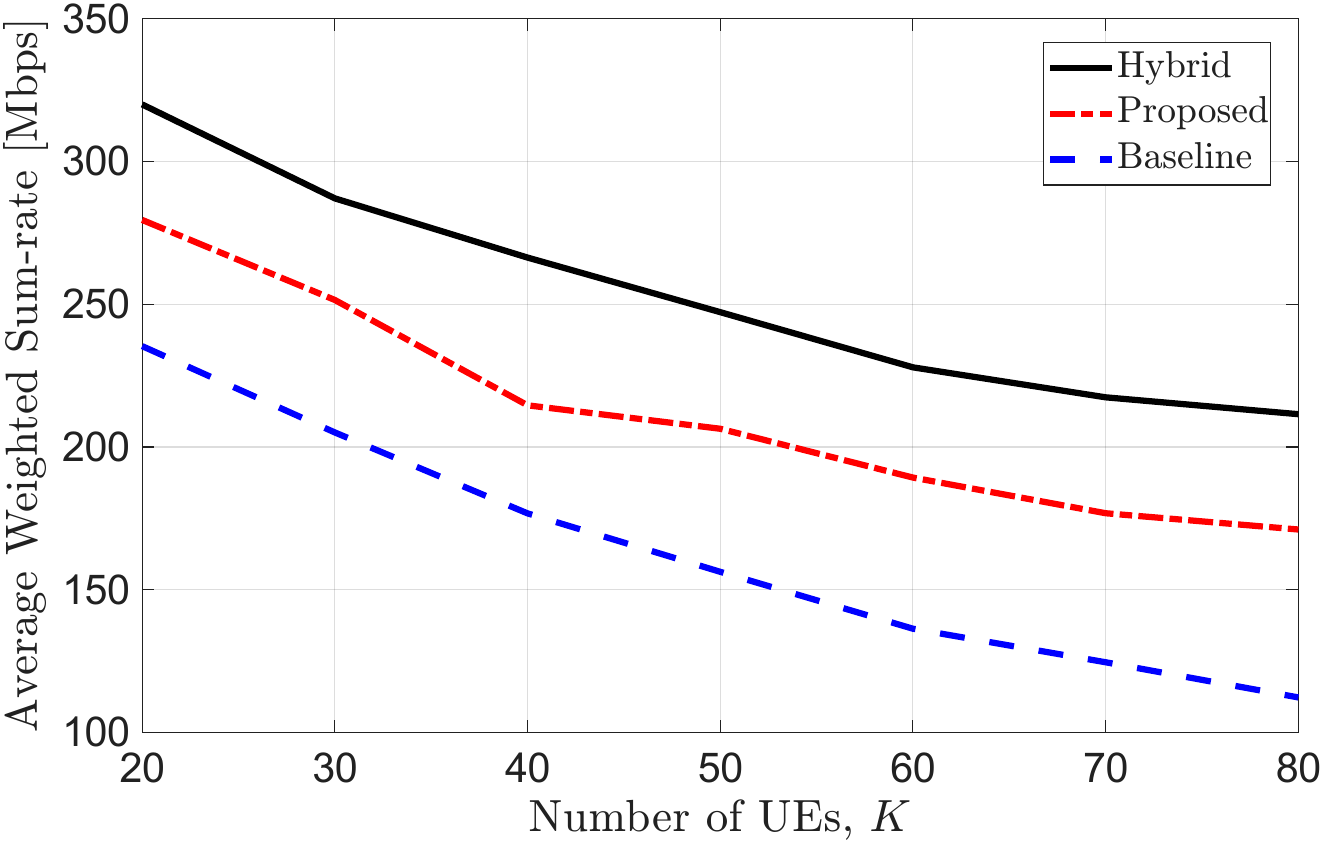}
		\caption{Average weighted sum-rate performance of different schemes for varying UE $K$ with $ \tau_p = 10 $ and $M = 100$ APs.}
		\label{Average_sumrate}
	\end{figure}
	
	\vspace{-1mm}
	\subsubsection{Comparison of Average Success Rate Performance}
	\label{Comparison of Average Success rate Performance}
	
	Figs.~\ref{Success_embb} and~\ref{Success_urllc} illustrate the average success rate, i.e., the fraction of UEs satisfying their respective QoS constraints: minimum-rate requirement \eqref{c:min_rate_unified} for eMBB and delay constraint \eqref{c:delay_unified} for URLLC UEs. In case of eMBB UEs (Fig.~\ref{Success_embb}), the Proposed scheme achieves $140\%$ and $115\%$ higher success rates than Baseline and Hybrid; for URLLC UEs (Fig.~\ref{Success_urllc}), gains are $58\%$ and $43\%$, respectively. Interestingly, while the Hybrid scheme delivers higher weighted sum-rate (Fig.~\ref{Average_sumrate}), the Proposed scheme excels in QoS satisfaction, demonstrating an effective trade-off. The Hybrid scheme maximizes the aggregate weighted objective by concentrating resources on high-weight UEs, potentially under-serving others below their QoS thresholds. Conversely, the Proposed scheme embeds QoS requirements directly into the $\zeta_k$ and $\Xi_{k,m}$ metrics (Algorithms~\ref{alg:greedy_knapsack} and~\ref{alg:greedy_association_simple}, respectively), factoring in both constraint satisfaction and sum-rate maximization during allocation. This yields a more balanced resource distribution, making the Proposed scheme  preferable when high success rates are critical despite modest sum-rate sacrifice. These URLLC success rates could be further improved by augmenting the framework with priority-based admission control, granting URLLC UEs preferential admission over eMBB UEs when resources are scarce. Such integration would enable the proposed bandwidth allocation to redistribute resources more effectively, enhancing URLLC reliability while preserving inter-slice fairness. Finally, examining the constraint characteristics, the substantially higher eMBB success rates reflect constraint stringency differences: the heterogeneity in eMBB rate requirements (for premium- and standard-tiers UEs) allows more flexibility in resource redistribution, whereas URLLC's uniformly stringent delay constraints offer limited room for adaptation.     
	
	\begin{figure}[tb]	
		\centering
		\includegraphics[scale=0.2]{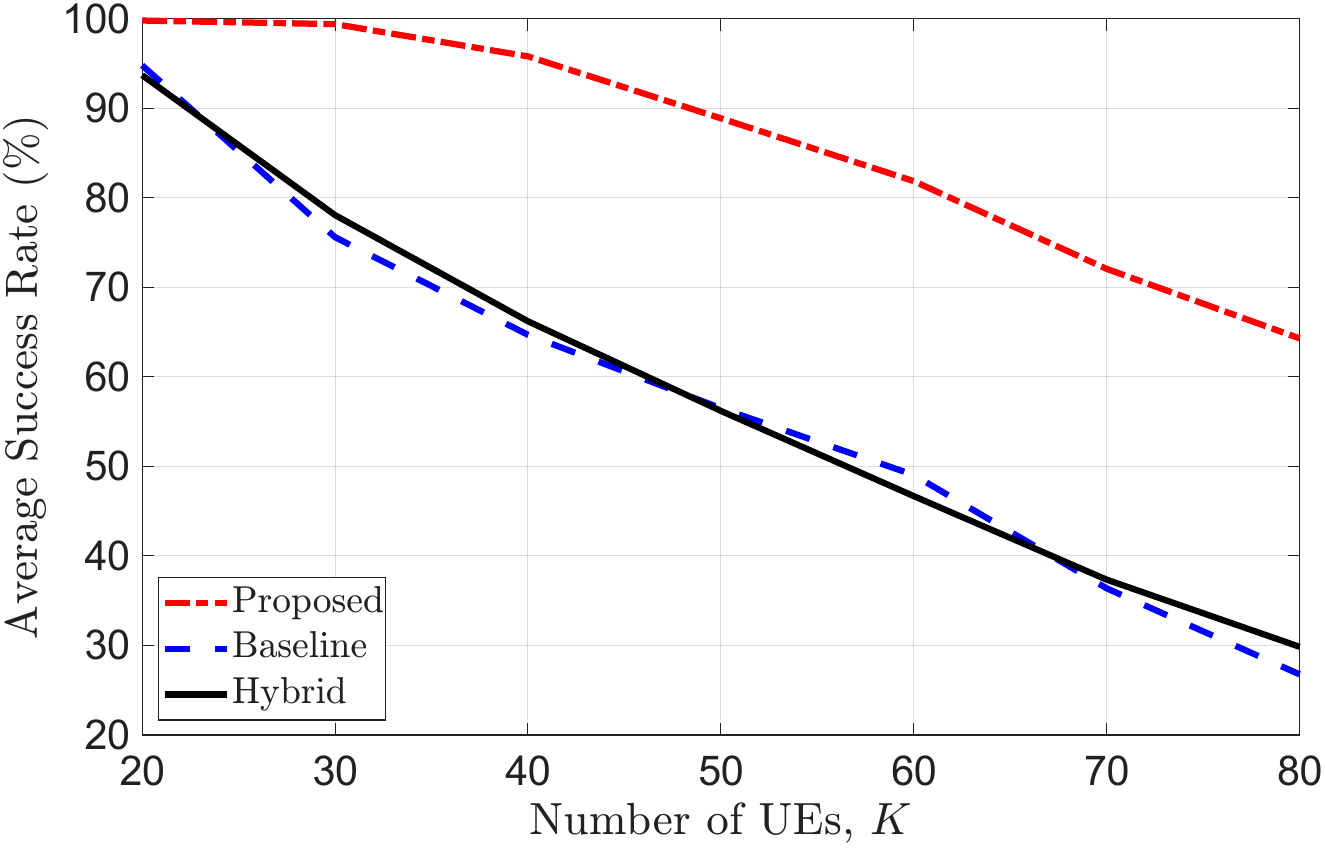}
		\caption{Average success rate performance of different schemes for eMBB UEs with $ \tau_p = 10 $ and $M = 100$ APs.}
		\label{Success_embb}
	\end{figure}
	
	\begin{figure}[tb]	
		\centering
		\includegraphics[scale=0.2]{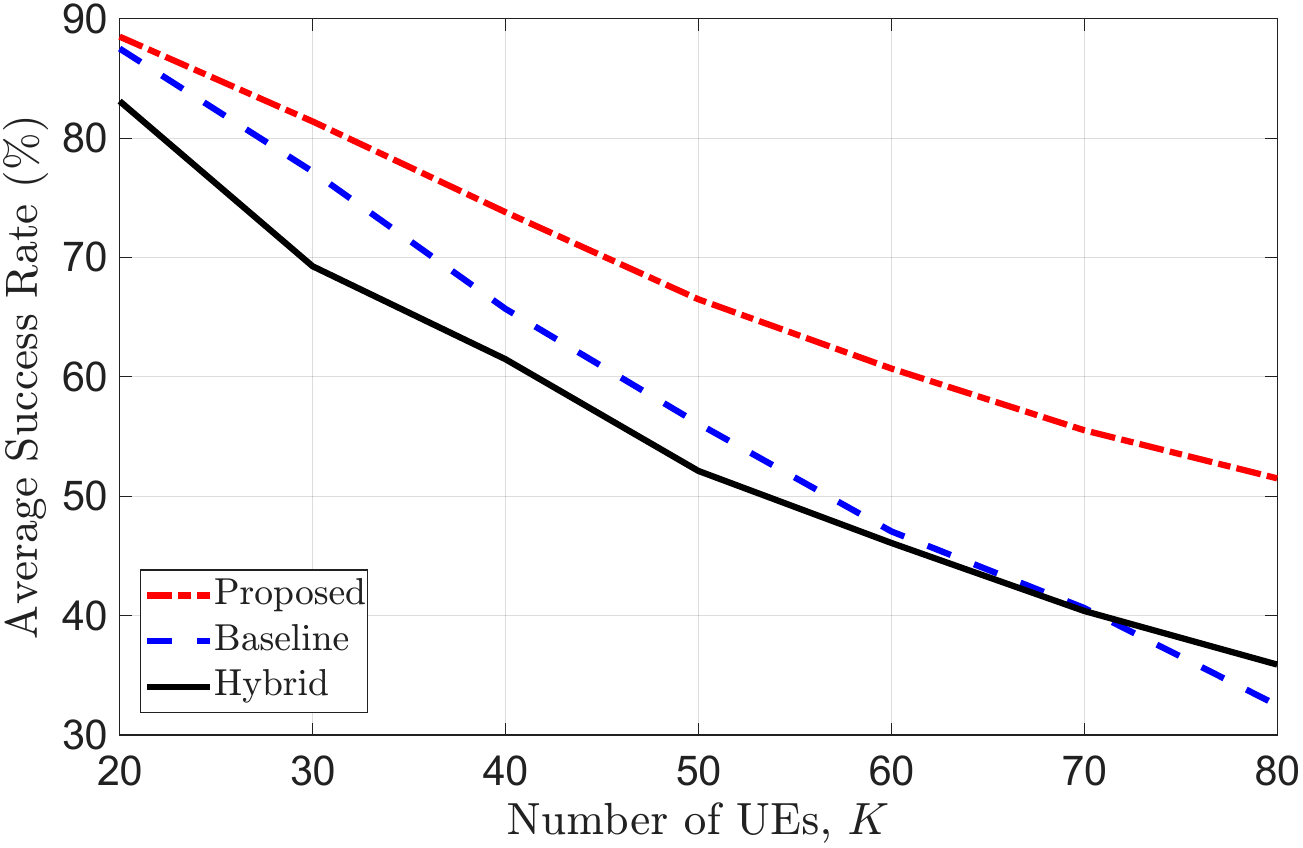}
		\caption{Average success rate performance of different schemes for URLLC UEs with $ \tau_p = 10 $ and $M = 100$ APs.}
		\label{Success_urllc}
	\end{figure}
	
	\subsubsection{Comparison of Computational Complexity Performance}
	\label{Comparison of Computational Complexity Performance}
	
	Fig.~\ref{runtime} presents the average runtime comparison between the Hybrid and Proposed schemes. As discussed in Section~\ref{Comparison of Average Weighted Sum-rate Performance}, the Proposed scheme exhibits a substantially lower computational complexity, reduced by up to two orders in $K$ compared to the Hybrid scheme. When the CVX solver operates successfully, the Proposed scheme achieves up to $97\%$ lower runtime for lower UE densities. However, at higher UE densities, both schemes exhibit nearly identical runtimes, implying that the Hybrid scheme’s CVX optimization is bypassed due to feasibility violations under Proposition~\ref{prop:feasibility}. In contrast, the Proposed scheme maintains consistent low complexity across all loads. Despite its lightweight computation, the Proposed scheme maintains superior QoS satisfaction and only a modest loss in weighted sum-rate compared to the Hybrid scheme. This highlights the Proposed scheme's overall computational efficiency and practicality for QoS-critical, resource-constrained large-scale deployments.	
	\begin{figure}[tb]	
		\centering
		\includegraphics[scale=0.2]{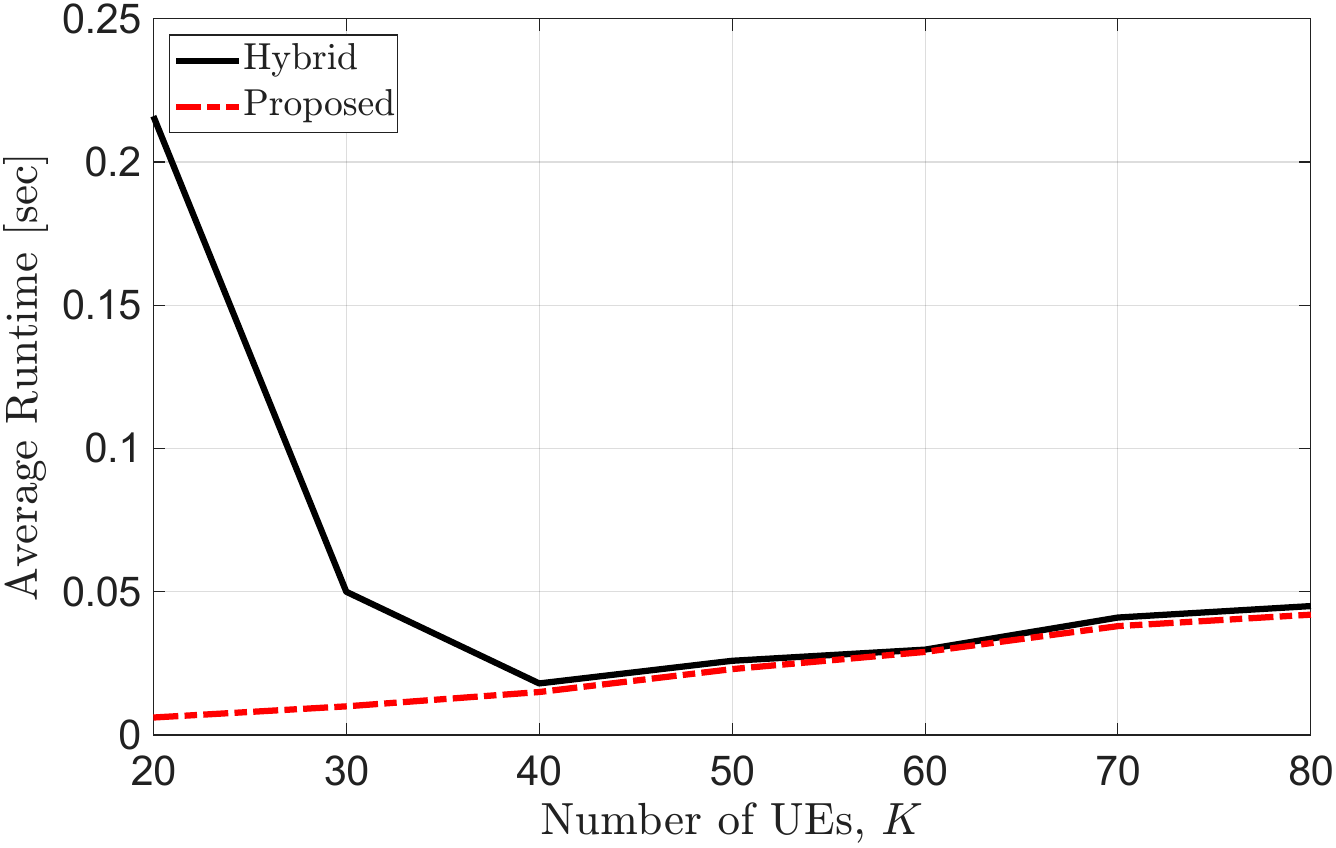}
		\caption{Average runtime of Hybrid and Proposed schemes for varying UE $K$ with $ \tau_p = 10 $ and $M = 100$ APs.}
		\label{runtime}
	\end{figure}
	\vspace{-2mm}	
	\section{Conclusion}
	In this paper, we present a computationally efficient framework for jointly optimizing per-UE bandwidth allocation and UE–AP association assignment in the UL of a NS-enabled CF mMIMO system, aiming to maximize the weighted sum-rate while satisfying heterogeneous QoS requirements for eMBB and URLLC UEs. Owing to the NP-hard nature of the formulated problem, it is decomposed into two tractable sub-problems, each addressed by an efficient heuristic scheme designed to ensure fair QoS distribution even under minimum-bandwidth constraint violation scenarios. Numerical results demonstrate that the proposed scheme achieves a favorable balance between computational efficiency, weighted sum-rate performance, and QoS satisfaction, making it a robust choice for QoS-critical network deployments. While this work focuses on maintaining service for all UEs through fair QoS distribution under resource scarcity, future extensions could explore admission control strategies to further enhance strict QoS guarantees for latency-critical applications.
	
	\vspace{-2mm}
	\appendix
	\section{Appendix A}
	\subsection{Proof of Theorem \ref{thm:problem_classification}}
	\label{Proof of Theorem 1}
	Problem \textbf{P0} contains $KM$ binary association variables $a_{k,m} \in \{0,1\}$ and $K$ continuous bandwidth allocation variables $b_k \in \mathbb{R}_+$, so it is a mixed-integer program. Nonconvexity follows from the objective’s dependence on the quadratic/bilinear terms (e.g., $(\sum_m a_{k,m}\gamma_{k,m})^2$ and $R_k=\left(1 - \frac{\tau_p}{\tau_c} \right)b_k\log_2(1+\mathrm{SINR}_k(\mathbf{A}))$), together with the discrete feasible set $\{0,1\}^{KM}$.
	
	For hardness, consider a restricted instance where bandwidths $b_k$ are fixed and QoS constraints are removed. Further, choose channel parameters (or assume negligible inter-UE interference) so that each $\mathrm{SINR}_k$ depends only on the AP assigned to UE $k$; then the objective reduces to maximizing $\sum_{k} w_k \log_2(1+\mathrm{SINR}_k(\mathbf{A}))$ which depends only on the assignment variables. Under this construction \textbf{P0} subsumes the generalized assignment problem (GAP): assigning items (UEs) to bins (APs) with capacity constraints and profits, which is NP-hard \cite{martello1990knapsack}. Hence, by reduction from GAP, \textbf{P0} is NP-hard.
	
	\vspace{-2mm}
	\bibliography{Jan31_2022}
	\bibliographystyle{ieeetr}
	
\end{document}